\documentclass[a4paper,USenglish]{lipics-v2016}

\usepackage{etex}
\reserveinserts{28}

\usepackage{xpunctuate}
\usepackage[all]{foreign} 
\usepackage[full,small]{complexity}

\usepackage[longend,vlined,linesnumbered]{algorithm2e}

\usepackage{lineno} 
\usepackage{graphicx}
\usepackage{epstopdf}
\usepackage{capt-of}
\usepackage{subcaption}
\usepackage{adjustbox}
\usepackage{scrextend}
\usepackage{enumerate}
\usepackage{enumitem}

\newcommand{\Problem}[1]{\textsc{#1}\xspace}
\newcommand{\weightset}{\ensuremath{\mathcal{U}}}
\newcommand{\pweights}{\ensuremath{\mathbf{w}}}

\newcommand{\XHS}{\Problem{X3HS}}
\newcommand{\XHSfull}{\Problem{Exact 3-Hitting Set}}

\newcommand{\FD}{\Problem{Flow Decomposition}}

\newcommand{\kFD}{\Problem{$k$-FD}}
\newcommand{\kFDfull}{\Problem{$k$-Flow Decomposition}}

\newcommand{\PPFD}{\Problem{$k$FWA}}
\newcommand{\PPFDfull}{\Problem{$k$-Flow Weight Assignment}}

\newcommand{\WPFDfull}{\Problem{$k$-Flow Routing}}

\newcommand{\WPUFD}{\Problem{\weightset-$k$FR}}
\newcommand{\WPUFDfull}{\Problem{\weightset-$k$-Flow Routing}}

\newcommand{\WPxFD}[1]{\Problem{$\{#1\}$-$k$FR}}
\newcommand{\WPxFDfull}[1]{\Problem{$\{#1\}$-$k$-Flow Routing}}

\newcommand{\UFD}{\Problem{\weightset-$k$FD}}
\newcommand{\UFDfull}{\Problem{\weightset-$k$-Flow Decomposition}}

\newcommand{\xFD}[1]{\Problem{$\{#1\}$-$k$FD}}

\newcommand{\probgfk}{($G,f,k$)\xspace}

\newcommand{\sdt}{$s$-$t$}
\newcommand{\stdag}{\sdt--DAG\xspace}

\newcommand{\stpath}{\sdt--path\xspace}
\newcommand{\stpaths}{\sdt--paths\xspace}
\newcommand{\TOcut}{top-cut\xspace}
\newcommand{\TOcuts}{top-cuts\xspace}

\newcommand{\maxf}{\ensuremath{\lambda}}
\newcommand{\toboggan}{\textsf{Toboggan}\xspace}
\newcommand{\catfish}{\textsf{Catfish}\xspace}

\newcommand{\zebra}{\texttt{zebrafish}\xspace}
\newcommand{\mouse}{\texttt{mouse}\xspace}
\newcommand{\human}{\texttt{human}\xspace}
\newcommand{\salmon}{\texttt{human-salmon}\xspace}

\newcommand{\preparaspace}{\vspace{0.07in}}

\theoremstyle{plain}
\newtheorem{conjecture}[theorem]{Conjecture}
\newtheorem{proposition}[theorem]{Proposition}

\def\N{\mathbb{N}} 

\def\Z{\mathbb{Z}}  

\DeclareRobustCommand{\stirling}{\genfrac\{\}{0pt}{}}

\newenvironment{tightcenter}
 {\parskip=0pt\par\nopagebreak\centering}
 {\par\noindent\ignorespacesafterend}

\usepackage{tikz}
\usetikzlibrary{calc}
\usepackage{xspace}
\usepackage{xargs}
\usepackage{xifthen}

\usepackage{framed}

\usepackage{ctable}
\newlength{\RoundedBoxWidth}
\newsavebox{\GrayRoundedBox}
\newenvironment{GrayBox}[1]%
   {\setlength{\RoundedBoxWidth}{\textwidth-4.5ex}
    \def\boxheading{#1}
    \begin{lrbox}{\GrayRoundedBox}
       \begin{minipage}{\RoundedBoxWidth}%
   }{%
       \end{minipage}
    \end{lrbox}%
    \begin{tightcenter}%
    \begin{tikzpicture}%
       \node(Text)[draw=black!20,fill=white,rounded corners,%
             inner sep=2ex,text width=\RoundedBoxWidth]%
             {\usebox{\GrayRoundedBox}};
        \coordinate(x) at (current bounding box.north west);
        \node [draw=white,rectangle,inner sep=3pt,anchor=north west,fill=white]
        at ($(x)+(10.5pt,.75em)$) {\boxheading};
    \end{tikzpicture}
    \end{tightcenter}\vspace{0pt}%
    \ignorespacesafterend
}

\newenvironment{problem}[2][]{\noindent\ignorespaces%
                                \FrameSep=6pt%
                                \parindent=0pt%
                \vspace*{-.5em}
                \ifthenelse{\isempty{#1}}{%
                  \begin{GrayBox}{\textsc{#2}}%
                }{%
                  \begin{GrayBox}{\textsc{#2} parametrised by~{#1}}%
                }
                \newcommand\Prob{Problem:}%
                \newcommand\Input{Input:}%
                \begin{tabular*}{\textwidth}{@{\hspace{.5em}} >{\itshape} p{1.2cm} p{0.82\textwidth} @{}}%
            }{
                \end{tabular*}%
                \end{GrayBox}%
                \vspace*{-.5em}
                \ignorespacesafterend
            }

\usepackage{graphicx}
\graphicspath{{./figures/}}

\def\paragraph#1{\par\textbf{#1} \ignorespaces}

\usepackage{url}
\makeatletter
\def\Url@twoslashes{\mathchar`\/\@ifnextchar/{\kern-1pt}{}}
\g@addto@macro\UrlSpecials{\do\/{\Url@twoslashes}}
\makeatother

\title{A practical fpt algorithm for \FD and transcript assembly}
\titlerunning{A practical fpt algorithm for \FD and transcript assembly} 

\author[1]{Kyle Kloster}
\author[2]{Philipp Kuinke}
\author[1]{Michael P. O'Brien}
\author[1]{Felix Reidl}
\author[2]{Fernando S\'anchez Villaamil}
\author[1]{Blair~D.~Sullivan}
\author[1]{Andrew~van~der~Poel}
\affil[1]{North Carolina State University, USA\\
  \texttt{(kakloste|mpobrie3|blair\_sullivan|ajvande4)@ncsu.edu, felix.reidl@gmail.com}}
\affil[2]{RWTH Aachen University, Germany\\
  \texttt{(kuinke|fernando.sanchez)@cs.rwth-aachen.de}}
\authorrunning{\small K. Kloster, P. Kuinke, M.P. O'Brien, %
              F. Reidl, F. S\'anchez Villaamil,%
              B. D. Sullivan, A. van der Poel}

\Copyright{Kyle Kloster, Philipp Kuinke, Michael O'Brien, %
           Felix Reidl, Fernando S\'anchez Villaamil,\\ %
           \phantom{Under Logo} Blair D. Sullivan, Andrew van der Poel}


\usepackage{verbatim}
\begin{document}

\maketitle

\begin{abstract}
  The \FD problem, which asks for the smallest set of weighted paths that
``covers'' a flow on a DAG, has recently been used as an important computational
step in transcript assembly.
We prove the problem is in \FPT~when parameterized by the number of paths by giving a practical
linear fpt algorithm.
Further, we implement and engineer a \FD solver based on this algorithm,
and evaluate its performance on RNA-sequence data.
Crucially, our solver finds exact solutions
while achieving runtimes competitive with a state-of-the-art heuristic.
Finally, we contextualize our design choices with two hardness results
related to preprocessing and weight recovery.
Specifically, \kFDfull does not admit polynomial kernels under standard complexity assumptions,
and the related problem of assigning (known) weights to a given set of paths is \NP-hard.

\end{abstract}


\section{Introduction}\label{sec:intro}


We study the problem \FD~\cite{HowToSplit, mumey2015parity, Catfish, greedy}
under the paradigm of parameterized complexity~\cite{downey2012parameterized}.
Motivated by the principles of algorithm engineering, we design and implement
an fpt algorithm that not only solves the problem exactly but also has a run-
time that is competitive with heuristics~\cite{Catfish, greedy}. Furthermore,
we characterize several key aspects of the problem's complexity.
Decomposing flows is the central algorithmic problem in a recent method for
analyzing high-throughput transcriptomic sequencing data, and we benchmark
our implementation on data from this use case~\cite{Catfish}.

\FD asks for the minimum number of weighted paths necessary to exactly cover a
flow on a directed acyclic graph with a unique source~$s$ and sink~$t$ (an
\stdag). More precisely, we say a set of \stpaths $P = p_1,\dots, p_k$ and
corresponding weights $\pweights = (w_1,\dots, w_k)$ are a \emph{flow
decomposition} of an \stdag $G$ with flow $f$ if
$$
  f(a) = \sum_{i=1}^k w_i\cdot\mathbf{1}_{p_i}(a)
$$
for every arc $a$ in $G$, where $\mathbf{1}_{p_i}(a)$ is the indicator function whose output is 1 if $a\in p_i$, and $0$ otherwise.
Specifically, in this paper, we are concerned with \kFDfull,
which uses the natural parameter of the number of paths.

\begin{problem}{$k$-Flow Decomposition (\kFD)}
    \Input & \probgfk with $G$ an \stdag, $f$ a flow on $G$, and $k$ a positive integer.\\
    \Prob  & Is there an integral \emph{flow decomposition} of $(G,f)$ using at most $k$ paths?
\end{problem}

\noindent
Shao and Kingsford
recently used flow decompositions to assemble unknown transcripts from
observed subsequences of RNA called exons~\cite{Catfish}. In this context, the
input graph is a \emph{splice graph}, which is constructed in one of two ways:
either by contracting induced paths in a de Bruijn graph (de novo assembly) or
by connecting exons inferred from an alignment process that co-occur in reads
(reference-based). In both constructions, every vertex in the graph corresponds to a (partial) exon being transcribed and every arc indicates that two exons appear in sequence in one of the reads; the arcs are weighted based on abundance (frequency) in the data.
If we affix a dummy source $s$ and a dummy sink $t$ to the vertices of in- and out-degree zero, respectively, the resulting labeling forms a flow\footnote{This assumes an idealized scenario with constant read coverage, but methods exist to rectify noisy data~\cite{tomescu2013novel}.} from $s$ to $t$.
As such, a transcript corresponds to a weighted \stpath and we aim to recover the collection
of transcripts by decomposing the flow into a minimal number of paths in the
above sense.

A precursory investigation of the data used by Shao and Kingsford~\cite{Catfish} to evaluate the performance of \FD heuristics in RNA transcript assembly left us with three guiding observations:
\vspace{-0.075in}
\begin{enumerate}
  \item 97\% of instances admit decompositions
        into fewer than 8 paths. Our algorithm should thus exploit
        the \textbf{small natural parameter}.
  \item The data set contains over 17 million mostly small instances. Our
        implementation should therefore be able to handle a \textbf{large throughput}.
  \item The flow decompositions computed by our implementation should
        reliably recover the \textbf{domain-specific solution} i.e.~the true transcripts.
\end{enumerate}
\vspace{-0.075in}
The first point corroborates the parameterized approach. Using dynamic
programming over a suitable graph decomposition, a common algorithm design
technique from parameterized complexity, we show the problem can be solved in
linear-fpt time (Section~\ref{sec:algo}):

\begin{theorem}\label{mainthm}
  There is an $2^{O(k^2)} (n + \maxf)$ algorithm for solving \kFDfull,
  where $\maxf$ is the logarithm of the largest flow value.
\end{theorem}

\noindent
To address the second point, we implement a \FD solver,
\toboggan~\cite{toboggan}, based on this algorithm and compare it with the
state-of-the-art heuristic \catfish (Section~\ref{sec:experiments}). Our
results show that \toboggan's running time is comparable to \catfish and thus
suitable for high-throughput applications. With respect to the third point,
using \FD for assembly implicitly assumes that the true transcripts
correspond to paths in a minimum-size flow decomposition.  Prior work focuses
on heuristics which cannot be used to evaluate the validity of this assumption.
With \toboggan, we validate that minimum-size flow decompositions accurately
recover transcripts in most instances (Section~\ref{sec:exp:validation}).

\toboggan incorporates several heuristic improvements
(discussed in Section~\ref{sec:implement}) to keep the running time and memory consumption
of the core dynamic programming routine as small as possible.
In particular, the simple preprocessing we employ is highly successful in solving many instances directly.
A provable guarantee for this preprocessing in the form of a small kernel, however, is unlikely: we
show that unless $\NP$ is contained in $\coNP / \emph{poly}$,
\kFD does not admit a polynomial kernel (Section~\ref{sec:no-poly-kernel}).
Further, our experimental evaluation hints at the following conjecture: for a given
decomposition of an \stdag into a minimum number of paths, there is a unique assignment of
weights to those paths to achieve a flow-decomposition.
If the conjecture holds, it implies a tighter bound on the running time of the algorithm in Theorem~\ref{mainthm}.
Despite this, we prove the more general `weight recovery' problem is \NP-hard given a decomposition into an arbitrary number of paths (Section~\ref{sec:exact3}).

\section{Preliminaries}
\label{sec:prelim}
\paragraph{Related work.}
\FD is known to be \NP-complete, even when all flow values are in $\{1,2,4\}$, and does
not admit a PTAS~\cite{HowToSplit,greedy}. The best known approximation algorithm for the problem is
based on parity-balancing path flows~\cite{mumey2015parity}, and guarantees an
approximation ratio of $\maxf L^{\maxf}$ with a running-time of
$O(\maxf|V|\cdot|E|^2)$, where $L$ is the length of the longest \stpath, and
$\maxf$ is the logarithm of the maximum flow value.
The variant in which the decomposition approximates the original flow values to
minimize a specified penalty function was studied in~\cite{tomescu-etal}, and
the authors give an fpt algorithm and FPTAS parameterized by $k$, the largest \sdt-cut, and the maximum flow value.

Due to its practical use for sequencing data, much of the prior work on \FD
has focused on fast heuristics, including two greedy algorithms which
iteratively add the remaining path which is shortest (\textsf{greedy-length})
or of maximum possible weight (\textsf{greedy-width})~\cite{greedy}.
Both heuristics produce
solutions with at most $k = |E| - |V| + 2$ paths and have variants~\cite{HowToSplit} which decompose all but an
$\varepsilon$-fraction of the flow using at most $O(1/\varepsilon^2)$ times
the minimum number of paths, for any $\varepsilon > 0$.
Historically, \textsf{greedy-width} has provided the best performance~\cite{HowToSplit,
mumey2015parity, greedy}, but the recent heuristic \catfish~\cite{Catfish}
showed significant improvements over \textsf{greedy-width} in accuracy and runtime.
\preparaspace
\paragraph{Notation.}  Given a directed acylic graph (DAG), $G = (V,A)$, we say $G$ is an
\stdag if $G$ has a single source, $s$, and a single sink, $t$.
We denote by $A^+(v)$ the set of out-arcs of a node $v$,
and by $A^-(v)$ the in-arcs. For a set of nodes, $S$,
we define $A^+(S) = \{vu \mid v\in S, u\notin S, vu \in A^+(v)\}$.
If $f$ is a flow on $G$, we write $f(a)$ for the flow value on
an arc $a$ and $F$ for the total flow (the sum of flow values on the arcs in $A^+(s)$).
\preparaspace
\paragraph{Terminology.}
Given a DAG, a \textit{topological ordering} on the nodes is a labeling such that every arc is directed from a node with a smaller label to one with a larger label.
We label the nodes of an \stdag~$G$ as $v_1,\dots, v_n$ corresponding to an arbitrary, fixed topological ordering of~$G$; this implies that $s = v_1$ and $t = v_n$.
We further define the sets~$S_i = \{v_j\mid j\leq i\}$ based on the ordering. We refer to the \sdt--cuts $A^+(S_i)$ as \emph{topological ordering cuts} (\TOcuts).
Our algorithm for computing a flow decomposition will ultimately depend on tracking how the paths cross \TOcuts, a notion we refer to as a \emph{routing}.


\begin{definition}[Routings and extensions]\label{def:extension}
  A surjective function $g \colon [k] \to A^+(S_i)$ is a \emph{routing out of $S_i$}.
  A routing $g' \colon [k] \to A^+(S_i)$ is an \emph{extension} of $g \colon [k] \to A^+(S_{i-1})$ if for each $j\in [k]$,
  $g'(j) = xy$ implies $g(j) = xy$ if $x\neq v_i$, and
  $g(j) = zx$ for some $z\in S_{i-1}$ if $x = v_i$.
\end{definition}

\noindent
In other words, an extension of a routing takes all the integers that map to
in-arcs of $v_i$ and maps them instead to the out-arcs of $v_i$ (because it is
surjective) while leaving the rest of the mapping unchanged. It will be important in our analysis
to differentiate arcs that occur in multiple paths from arcs that appear on
only a single path: we say paths $p$ and $p'$ \emph{coincide} on arc $a$ if
$a\in p$ and $a\in p'$. An arc $a$ is \emph{saturated} by a path~$p$ if~$p$ is
the only path for which $a\in p$.
\preparaspace
\paragraph{Parameterized complexity.}
A parameterized decision problem~$\Pi \in \Sigma^* \times \mathbb N$ is
\emph{fixed parameter tractable} if there exists an algorithm that decides it
in time~$g(k) \cdot n^c$ for some computable function~$g$ and a constant~$c$.
When~$c=1$, we call such an algorithm \emph{linear fpt}.
A \emph{polynomial kernel} is an algorithm that transforms, in polynomial time,
an instance~$(I,k) \in \Sigma^* \times \mathbb N$ of~$\Pi$ into an
equivalent instance~$(I',k')$ with~$|I'|,k' \leq k^{O(1)}$, meaning
that~$(I,k) \in \Pi$ if and only if~$(I',k') \in \Pi$.
Some more advanced machinery pertaining to kernelization lower bounds will be
defined in Section~\ref{sec:no-poly-kernel}.




\section{A linear fpt algorithm for \kFD}
\label{sec:algo}

\noindent
We solve \kFD via dynamic programming over the topological ordering:
we enumerate all ways of routing $k$ \stpaths over each \TOcut~$A^+(S_i)$.
Each such routing determines a set of constraints for the path weights,
which we encode in linear systems. For example, if
paths~$p_1$ and $p_2$ are routed over an arc~$a$, we add the constraint
$w_1 + w_2 = f(a)$ to our system.

Concretely, we keep a table $T_i$ for~$0 \leq i \leq n$ whose entries
are sets of pairs $(g,L)$, where $g\colon [k] \to A^+(S_i)$ is a
routing of the paths out of $S_i$ and $L$ is a system of linear equations that
encodes the known path weight constraints.
In particular, for every arc $a \in A^+(S_i)$, we write
$g^{-1}(a)$ for the set of \stpaths that are routed over $a$.
Each system of linear equations, $L$, is of the form $\mathbf{Aw} =
\mathbf{b}$, where $\mathbf{A}$ is a binary matrix with~$k$ columns,
$\mathbf{w}$ is the solution weight vector, and~$\mathbf{b}$ is a vector
containing values of~$f$. Each row $r$ of $\mathbf{A}$ corresponds to
an arc $a_r$ and encodes the constraint that
the weights of paths routed over~$a_r$ sum up to~$f(a_r)$.
The $j$th entry of row $r$ is 1 if and only if
path~$p_j$ was routed over $a_r$, and then~$f(a_r)$ equals the $r$th
entry of $\mathbf{b}$. Figure~\ref{fig:table_entry} illustrates an example of an entry~$(g,L)$.

\begin{figure}
    \centering
    \includegraphics[scale=.85]{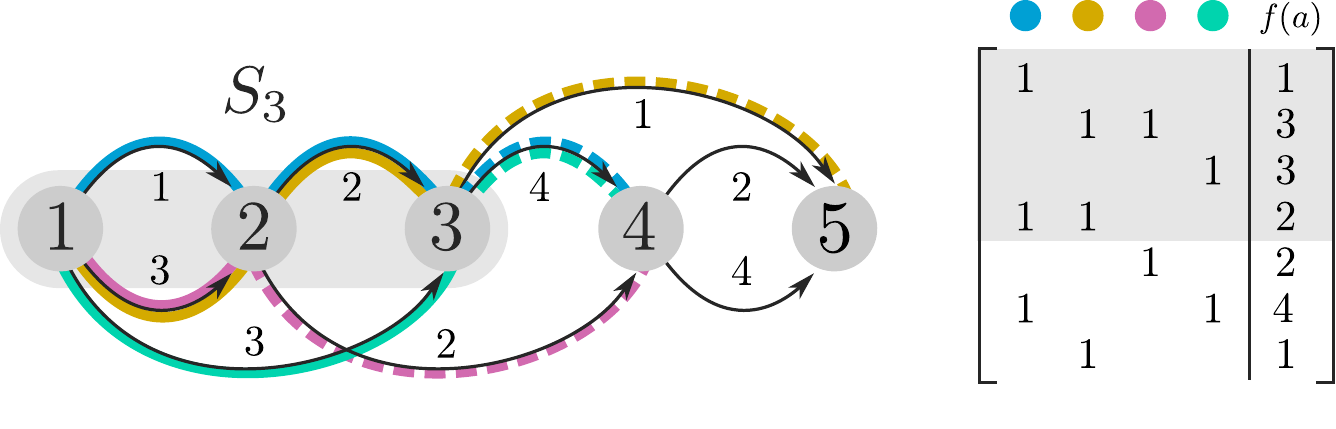}
    \caption{
        An entry of the table $T_3$. The routing~$g$ out of~$S_3$ (dashed lines) is an extension of the previous routings (solid paths).
        Each row in the constraint system~$L$ on the right corresponds to an arc; those shaded in gray are from arcs inside~$S_3$, and those in white
        come from $g$.
    }
    \label{fig:table_entry}
\end{figure}

We now describe how the dynamic programming tables
are constructed. For ease of
description, we augment $G$ to have two additional dummy arcs $a_s$ and $a_t$,
where $a_s$ is an in-arc of $s$ and $a_t$ is an out-arc of $t$ and $f(a_s) =
f(a_t) = F$. We begin with table $T_0$ that has a single entry. The routing of
this entry routes all paths over $a_s$ and the linear system has a single row
that constrains all path weights to sum to the total flow value.

For $i> 0$, we construct the table $T_i$ from $T_{i-1}$. Conceptually, we
``visit'' node $v_i$ and ``push'' all paths routed over its
in-arcs to be routed over its out-arcs.
Formally, we require the routing out of $S_i$ to be an
extension (Definition~\ref{def:extension}) of a routing out of $S_{i-1}$.
For each entry $(g,L) \in T_{i-1}$ we compute an entry $(g',L')$ of $T_i$ for
each extension $g'$ of $g$. For a given $g'$, we create $L'$ by adding
a row to $L$ for each arc $a \in A^+(v_i)$, encoding the constraint
$$
  \sum_{i \in g'^{-1}(a)} w_i = f(a).
$$
At the conclusion of the dynamic programming, all entries of the final table
$T_n$ will have the same routing, since $v_n=t$ has one (dummy) out-arc $a_t$.
This table allows us to decide whether there is a solution; $(G, f)$ is a yes-instance of \kFD if and only if some $L\in T_n$ has a solution
$\mathbf w$ whose entries are positive integers.
Pseudocode for our algorithm can be found in
Algorithm~\ref{alg:main} in Appendix~\ref{sec:pseudo}.

\begin{lemma}\label{lem:correctness}
  An integer vector $\pweights$ is a solution to $L \in T_n$ if and only
  if there is a set of \stpaths $P$ such that $(P, \pweights)$ is a flow
  decomposition of $(G, f)$.
\end{lemma}
\begin{proof}
    We first prove the forward direction.
    If we keep backpointers in our DP tables i.e.\ pointers from entry $(g_i, L_i)$ to the entry in the previous table $(g_{i-1},L_{i-1})$ for which $g_{i}$ was the extension of $g_{i-1}$, we can obtain a sequence of routings $g_{n-1},\dots, g_1$ that correspond to backwards traversal of the backpointers.
    Let $h(j) = g_1(j), \dots, g_{n-1}(j)$.
    By Definition~\ref{def:extension}, removing the duplicate elements that appear in consecutive order from $h(j)$ yields a series of arcs that form an \stpath $p_j$.
    Because each routing requires every arc to have a path routed over it, the system $L$ contains constraints corresponding to each arc.
    Thus an integer solution $\pweights$ to $L$ corresponds to a weighting of $P = p_1,\dots, p_k$ such $(P, \pweights)$ is a flow decomposition.

    In the reverse direction, we first observe that the incidence of paths in $P$ on a \TOcut $A^+(S_i)$ corresponds to a routing out of $S_i$.
    Let $g_i$ be the corresponding routing out of $S_i$.
    For each node $v_i$, the paths routed over $A^-(v_i)$ must immediately be routed over an arc in $A^+(v_i)$, meaning $g_i$ is an extension of $g_{i-1}$.
    Because $(P, \pweights)$ is a flow decomposition, the paths routed over each arc will have weights summing to the flow value on that arc.
    Thus, any constraint in a linear system associated with the routings $g_i$ will have $\pweights$ as a solution.
\end{proof}

\noindent
To analyze the running time we now derive bounds on the dynamic programming table sizes.
First, we bound the number of possible routings in Lemma~\ref{lem:num_routings},
then we give an upper bound on the number of linear systems in Lemma~\ref{lem:linear-systems}.

\begin{lemma}\label{lem:num_routings}
 There are at most $\sqrt k \, \left(0.649 k\right)^k$ routings of $k$ paths over any \TOcut.
\end{lemma}
\begin{proof}
   A routing of $k$ paths over a \TOcut $C$ can be thought of as
   a partition of the paths into $\ell = |C|$ many non-empty sets along with a
   mapping of these sets to the arcs of $C$.
   We can assume that~$\ell \leq k$,
   since a routing of $k$ paths requires every \TOcut to have at most $k$ arcs.
   The number of ways to partition $k$ objects into $\ell$
   non-empty sets is $\stirling{k}{\ell}$, the Stirling number of
   the second kind, and there are $\ell!$ ways to assign each of
   these partitions to a specific arc.
   Thus, the number of routings is $\stirling{k}{\ell}\ell!$.
   To proceed, we use the upper bound~$\stirling{k}{\ell} \leq \frac{1}{2} {k
   \choose \ell} \ell^{k-\ell}$ due to Rennie and Dobson~\cite{StirlingBound}.
   We also make use of the tighter version of Stirling's approximation due to
   Robbins~\cite{TighterStirling}, which states that
  $$
     \sqrt{2\pi k} \left( \frac{k}{e} \right)^k e^{1/(12k+1)} \leq k! \leq \sqrt{2\pi k} \left( \frac{k}{e} \right)^k e^{1/12k}.
  $$
  Hence, we have the upper bound
  \begin{align*}
      \stirling{k}{\ell} \ell! &\leq \frac{1}{2} {k \choose \ell} \ell^{k-\ell} \ell!
                        = \frac{1}{2} \frac{k!}{(k-\ell)!} \frac{\ell^k}{\ell^\ell} \\
        &\leq \frac{\sqrt{k}}{2\sqrt{k-\ell}} \, \left(\frac{k}{e}\right)^k \left(\frac{e}{k-\ell}\right)^{k-\ell} e^{\frac{1}{12k} - \frac{1}{12(k-\ell)+1}}     \frac{\ell^k}{\ell^\ell} \\
        &\leq \sqrt{k} \, \frac{k^k \ell^{k-l}}{(k-\ell)^{k-\ell} e^\ell}.
  \end{align*}
  Letting~$\ell = \alpha k$ for~$\alpha \in (0,1)$, the above expression becomes
  $$
      \sqrt k \, \frac{(\alpha k)^{k - \alpha k}}{(k-\alpha k)^{k-\alpha k} e^{\alpha k}} k^k
      \leq \sqrt k \, \left(  \left(\frac{\alpha}{1-\alpha}\right)^{(1-\alpha)}
            e^{-\alpha} \right)^k k^k \leq \sqrt k \, (0.649 k)^k,
  $$
  where the constant $0.649$ can be derived numerically by maximizing
  $g(\alpha) = \left(\frac{\alpha}{1-\alpha}\right)^{(1-\alpha)} e^{-\alpha}$
  on the interval $[0,1]$.
\end{proof}

\begin{lemma}\label{lem:linear-systems}
    Each table has at most $\frac{4^{k^2}}{k! \, k^k}$ distinct linear
    systems.
\end{lemma}
\begin{proof}
  Without loss of generality, we can ensure each linear system $L$ has at most $k$ rows by removing linearly dependent rows.
  We note that because there are only $2^k$ subsets of weights, if $f$
  maps the arcs to more than $2^k$ unique flow values, there cannot be a flow decomposition of size $k$.
  Since the elements of~$\mathbf{b}$ can thus take on at most~$2^k$ many values,
  and~$\mathbf{A}$ contains binary rows of length~$k$, it follows that there are at most $4^k$ rows for~$L$.
  Accordingly, we can bound the number of possible
  linear systems by~$\binom{4^k}{k}$. By imposing an order on the rows
  we can remove a factor of~$1/k!$. Because~$k \leq \sqrt{4^k}$, we can apply
  the bound~$\binom{n}{k} \leq (n/k)^k$ which holds~\cite{lipton}
  for~$k \leq \sqrt{n}$; thus, the number of linear
  systems is at most~$4^{k^2}/(k! \, k^k)$.
\end{proof}

\noindent
Now that we have bounded the size of the dynamic programming tables,
we analyze the complexity of solving the linear systems in the final
table $T_n$. It has been shown that treating
linear systems as integer linear programs (ILPs)  produces integer solutions in fpt-time.\looseness-1

\begin{proposition}[\cite{lokshtanov2009integer}]\label{prop:ilp-fpt}
  Finding an integer solution to a given system~$L$ takes time at most
  $O(k^{2.5k+o(k)} |L|)$, where $|L|$ is the encoding size of the linear
  system.
\end{proposition}

\noindent
These results give the following upper bound on the runtime of
our algorithm, proving Theorem~\ref{mainthm}.

\begin{theorem}\label{thm:runtime}
  Algorithm~\ref{alg:main} solves \kFD
  in time $4^{k^2} k^{1.5k} k^{o(k)} 1.765^k \cdot (n + \lambda)$
  where $\lambda$ is the logarithm of the largest flow value of the input.
\end{theorem}
\begin{proof}
  The correctness of the algorithm was already proven in
  Lemma~\ref{lem:correctness}, so all that remains is to bound the running time.
  By Lemmas~\ref{lem:num_routings} and~\ref{lem:linear-systems}, the total
  number of elements in DP table~$T_i$ is bounded by
  $$
    \frac{4^{k^2}}{k! \, k^k} \cdot \sqrt k \, (0.649 k)^k = \sqrt k \frac{4^{k^2} 0.649^k}{k!}.
  $$
  Reducing the linear systems and  checking for consistency is polynomial in
  the size of the matrix (which is $k \times k$). Finally, we need
  to find a feasible solution among all the linear systems left after the last
  DP step. We apply Proposition~\ref{prop:ilp-fpt} to these systems, whose
  encoding size is bounded by~$k^{O(1)} \lambda$.
  We arrive at the desired running time of
  $$
    \frac{4^{k^2} k^{2.5k} 0.649^k }{k!} k^{o(k)} \cdot (n + \lambda)
    \leq 4^{k^2} k^{1.5k}  1.765^k k^{o(k)} \cdot (n + \lambda)
  $$
  where we use the well-known bound $\tfrac{k^k}{k!} \leq e^k$.
\end{proof}

\section{Implementation}
\label{sec:implement}

To establish that our exact algorithm for \kFD is a viable alternative to the
heuristics currently in use by the computational biology community, we
implemented Algorithm~\ref{alg:main} as the core of a \FD solver
\toboggan~\cite{toboggan}. The solver iterates over values of $k$
in increasing order until reaching a yes-instance of \kFD.
\toboggan also implements backtracking to recover the \stpaths.
Making \toboggan's runtime competitive with existing implementations of
state-of-the-art heuristics required non-trivial algorithm engineering.
Our improvements broadly fall into three categories: preprocessing, pruning,
and low-memory strategies for exploring the search space.
The remainder of this section describes the most noteworthy techniques implemented in \toboggan that are not captured by Algorithm~\ref{alg:main}.

\subsection{Preprocessing}\label{sec:preprocessing}

\toboggan implements two key preprocessing routines. The first
generates an equivalent instance on a simplified graph
using a series of arc contractions. The second calculates
lower bounds on feasible values of $k$ to reduce the number of calls
to the \kFD solver.
\preparaspace
\paragraph{Graph reduction.}
We first reduce the graph by contracting all arcs into/out of nodes of
in-/out-degree 1.
We prove in Lemma~\ref{lem:contracted} that given a flow decomposition of the contracted graph, we can efficiently recover a decomposition of the same size for the original graph.

\begin{lemma}\label{lem:contracted}
    Let $uv$ be an arc for which $|A^+(u)| = 1$ or $|A^-(v)| = 1$ and $G'$ the graph created by contracting $uv$.
    Then $(G', f)$ has a flow decomposition $(P', \pweights)$ of size $k$ iff $(G, f)$ has a flow decomposition $(P, \pweights)$ of size $k$.
	Moreover, we can construct $(P, \pweights)$ from $(P', \pweights)$ in polynomial time.
\end{lemma}
\begin{proof}
    We first note that if $G^R$ is the graph formed by reversing the directions of the arcs in $G$, any solution to $G$ can be transformed into a solution to $G^R$ by reversing the order of the paths and maintaining the same weights.
    Since this reversal is involutive, the correspondence between solutions to $G$ and $G^R$ is bijective, so it suffices to consider the case $|A^+(u)|=1$.

    Given $G$ with a node $u$ that has $|A^+(u)|=1$, let $G'$ be the graph $G$ with arc $uv$ contracted.
    Given a flow decomposition $(P, \pweights)$ of $(G, f)$, we construct a corresponding decomposition $(P', \pweights)$ for $(G', f)$ as follows.
    Every path $p\in P$ containing $u$ must have $v$ as the vertex succeeding $u$.
    Removing $u$ from each such $p$ will create a valid path in $G'$, since $A^-(u)$ becomes part of $A^-(v)$ in $G'$.
    Moreover, the incidence of paths on each other arc remains unchanged, so the solution covers the flows in $G'$.

    In the reverse direction, consider a labeling of the arcs of $G'$ such that we can distinguish among the in-arcs of $v$ those which exist in $G$ from those that result from contracting $uv$.
    Let $A_{\text{new}}$ be the set of latter arcs.
    We construct $P$ from $P'$ as follows.
    For each path $p\in P'$, if $p$ is routed over $xv\in A_{\text{new}}$ we modify $p$ to include the arcs $xu$ and $uv$, rather than the arc $xv$, and add the modified path to $P$.
    If no such arc lies in $p$, we simply add $p$ to $P$.

    As in the forward direction, it is clear that any arcs in $G$ without $u$ as an endpoint have their flow values covered by $P$, and that every path in $P$ is an \stpath in $G$.
    Since $xv\in A_{\text{new}}$ is derived from an arc $xu\in G$, if $Q'\subset P'$ are the paths routed over $xv$ and $Q$ is the set of paths corresponding to $Q'$ in $P$, then $Q$ exactly covers $xu$.
    Since every path in $P$ routed through $A^-(u)$ subsequently is routed over $uv$ and the flow values on the in-arcs of $u$ sum to $f(uv)$, $uv$ is also covered by $P$.
    Thus $(P, \pweights)$ is a valid solution to $(G,f)$.
\end{proof}
\paragraph{Lower bounds on $k$.}
To reduce the number of values of $k$ that \toboggan considers before reaching a yes-instance,
 we compute a lower bound on the optimal value of $k$.
One immediate lower bound is the size of the largest \TOcut;
we implement this in conjunction with the additional bound established in Lemma~\ref{lem:k_opt_bound}.
Intuitively, if two \TOcuts share few common flow values, many of the arcs must have multiple paths routed over them.

\begin{lemma}
    \label{lem:k_opt_bound}
    Given a flow $(G,f)$, let $C_1$ and $C_2$ be any two \TOcuts with $|C_1| \geq |C_2|$.
    Letting $\mathcal{F}(S)$ be the multiset of flow values occurring in $S$,
    set $F_1 = \mathcal{F}(C_1) \setminus \mathcal{F}(C_2)$ and $F_2 = \mathcal{F}(C_2) \setminus \mathcal{F}(C_1)$.
    If $(G, f)$ has a flow decomposition of size $k$, then
    \begin{equation}\label{eqn:lowerboundk}
    k \geq |\mathcal{F}(C_1) \cap \mathcal{F}(C_2)| + \tfrac{2}{3}(|F_1|+|F_2|),
  \end{equation}
    and this lower bound is tighter than the cutset size $|C_1|$ iff $|F_1| < 2|F_2|$.
\end{lemma}

\begin{proof}
  Suppose $k$ paths cross cutsets $C_1$ and $C_2$.
  For every flow value in $\mathcal{I} = \mathcal{F}(C_1) \cap \mathcal{F}(C_2)$, it is possible a single path saturates the arc with that flow value in both cuts.
  Consider the remaining $h = k - |\mathcal{I}|$ paths.
  To maximize the number of distinct flow values these $h$ paths produce, let each path saturate an arc in $C_1$, yielding $h$ values in $F_1$, and then route those $h$ paths in pairs across distinct arcs in $C_2$. This produces at most $h/2$ new flow values in $F_2$. This yields $(3/2)h \geq |F_1| + |F_2|$, and substituting for $h$ proves Inequality~\eqref{eqn:lowerboundk}.

  To prove the relationship between this lower bound and $|C_1|$ there are two cases.
  If $|F_1| \geq 2|F_2|$, then by substituting we can upper bound $(2/3)(|F_1|+|F_2|) \leq |F_1| = |C_1| - |\mathcal{I}| \leq |C_1|$.
  If instead $|F_1| < 2|F_2|$, substituting yields $|\mathcal{I}| + (2/3)(|F_1|+|F_2|) > |\mathcal{I}| + |F_1| = |C_1|$.
\end{proof}

\subsection{Search Space Strategies}\label{sec:control-flow}
To reduce the memory required by dynamic programming, our implementation diverges from the pseudocode of Algorithm~\ref{alg:main} in two ways.
Specifically, we solve a restricted weight variant of \kFD and use a separate phase to recover the \stpaths.
\preparaspace
\paragraph{Weight restriction.}
Rather than making one pass through the dynamic programming that infers the
solution weights from the linear systems, we employ a multi-pass strategy.
Each pass restricts the potential weight vector by fixing a
subset of its entries to weights chosen from flow values in the
input instance while leaving the remaining values as variables to
be determined by the DP routine. Since we observe that most instances in
the data set (see Section~\ref{sec:data_setup}) admit decompositions whose paths saturate at least one arc, we begin by enumerating potential weight-vectors
whose values are \emph{all} taken from flow values. If the DP routine finds
no solution for any of these vectors, we enumerate vectors in which exactly one
value is undetermined and all other entries are again taken from flow values,
\etc If all these attempts should fail, the final pass in our strategy will
leave all weights to be determined by the DP and hence is guaranteed to find a
solution, should one exist. For most instances, however, a solution vector
is quickly guessed, vastly reducing the number of candidate linear systems
in the DP and thus saving both time and memory.
\preparaspace
\paragraph{Path recovery.}
Computing the path weights requires storing only the current and previous dynamic programming tables in memory.
In contrast, recovering the paths requires storing \emph{all} dynamic programming tables.
For this reason, we first determine the weights $\pweights$ and then recover the paths by
running the dynamic programming again with weights restricted to $\pweights$.
This is equivalent to solving the \WPFDfull problem in Section~\ref{sec:no-poly-kernel}.

\subsection{Pruning}\label{sec:pruning}
Within the dynamic programming we employ a number of heuristics that help recognize algorithmic states that cannot lead to a solution.
\preparaspace
\paragraph{Weight bounds.}
We augment the weight constraints imposed by the linear systems with a set of routing-independent constraints.
These take the form of upper and lower bounds ($B_i$ and $b_i$, respectively) on the $i$th smallest weight $w_i$.

First, we compute $b_k$ by noting that for any \TOcut $C$ and any arc $a\in C$, only $k-|C|+1$ paths can be routed over $a$, i.e.\ $b_k \geq f(a)/(k-|C|+1)$.
Then, we compute $B_k$ by finding the largest weight of any \stpath. This can be done via a simple dynamic programming algorithm: for any node $v$, an $s$-$v$--path with weight $w$ requires a path of equal weight to an in-neighbor $u$ such that $w \leq f(uv)$.

To compute the bounds on the other weights, $i < k$, we let $B_1$ be the smallest flow value and $B_i$ be the
smallest\footnote{If no such flow value exists, set $B_i = B_k$.}
flow value greater than $\sum_{j=1}^{i-1} B_j$.
In other words, if $f(a) > \sum_{j=1}^{i-1} w_j$, then $a$ must be part of a path of weight at least $w_i$.
Finally, we use these upper bounds to derive the $b_i$s.
If the weights greater than $w_i$ sum to $W$, then by the pigeonhole principle
$w_i \geq (F-W)/i$, where $F$ is the total flow. Thus, $w_i \geq (F -
\sum_{j=i+1}^{k} B_j)/(i+1)$.

For simplicity of implementation, we force the values of our weight vector to appear in non-decreasing order.
In each dynamic programming table, for each linear system $L$ we run a linear program to see if there are any (rational) weights within the bounds that satisfy $L$.
If not, we remove the entry of the dynamic programming table containing $L$.
We also use these bounds before dynamic programming to catch partially determined weight vectors that will not be able to yield a solution.
For example, the predetermined vector\footnote{A $*$ indicates an undetermined value.} $\pweights = (1,3,*,*,5,11)$ is incompatible with a lower bound on the third entry greater than 5.

\preparaspace
\paragraph{Storing linear systems.}
Within the dynamic programming we store linear systems in row-reduced echelon form (RREF).
When a new row $r$ is introduced to the system, we perform Gaussian elimination to convert the new row to RREF, checking for linear dependence and inconsistency in $O(k^2)$ time.
The iterative row reduction also has the advantage of revealing determined path weights even if the system is not fully determined, and thus we can check whether any such values are not positive integers.
Furthermore, once the system is full rank, no additional computation needs to be done to recover the weights.

\section{Experimental Results}
\label{sec:experiments}

In this section we empirically evaluate the efficiency and solution quality of
\toboggan by comparing with the current state-of-the-art, \catfish, on
a corpus of simulated RNA sequencing data. Additionally, we
use \toboggan to theoretically validate the \kFD problem as a model for the
transcript assembly task.

\subsection{Experimental setup}\label{sec:data_setup}
Our experiments were run on a corpus of data used in previous experiments by
Shao and Kingsford~\cite{catfish-github,Catfish}. The full data set contains over 17 million instances, each generated by simulating RNA-seq reads from a transcriptome and building the \stdag and flow using the procedure described in Section~\ref{sec:intro}.
Two software packages were used to simulate reads:  \textsf{Flux Simulator}~\cite{flux-simulator}, which used transcriptomes from humans, mice, and zebrafish, and \textsf{Salmon}~\cite{salmon}, which only used human transcriptomes.
We will refer to these four subsets as \human, \mouse, \zebra, and \salmon, respectively.
The simulated reads allow us to know the true transcripts and thus the corresponding flow decomposition \textit{a priori}; we will refer to this
particular decomposition as the ``ground truth''.

In our three experimental analyses (model validation, algorithm timing, and solution-quality) we report results on just the three smaller datasets (\human, \mouse, \zebra).
We terminated \toboggan on any run whose weight computation took longer than
800 seconds on \zebra, and 50s on \human and \mouse\footnote{Under these constraints, \toboggan timed out on $5\,136$ of the 4M instances in these smaller data sets.}.
After performing more exhaustive computations on these smaller data sets, we then used the larger \salmon data set (13.3M) to add one more data point for the model validation experiment.
Because of the size of \salmon, we terminated \toboggan after only 1 second.
This smaller time limit resulted in \toboggan timing out on a significantly larger portion of \salmon than the other three ($3\%$ on \salmon compared to $0.1\%$ on the other data sets), which would bias the results on \salmon toward its simpler instances.
Because of this, we omit the \salmon results from the timing and solution-quality experiments.

For a small number of the 17M instances, the ground truth decomposition contains at least one path that appeared multiple times, corresponding to a subsequence repeated in different locations of the transcriptome.
The data provides no means for mathematically or biologically distinguishing these; thus, we aggregated the duplicated paths into a single path, summing their weights.
Additionally, we removed ``trivial'' instances in which the graph consisted of a single \stpath;
on such graphs \toboggan terminates during preprocessing without executing the \kFD algorithm described in this paper.
We remark that this is a departure from the experimental setup of \catfish, which included such graphs, explaining some slight differences in our statistics.
About $50\%$ of the 17M graphs in the corpus are trivial in this manner.
The number of non-trivial graphs for each species is summarized in Table~\ref{tab:validation}.
All experiments were run on a
dedicated system with an Intel i7-3770 (3.40GHz, 8 cores), an 8192 KB cache, and 32 GB of RAM.

\subsection{Benchmarking}\label{sec:exp:benchmark}
We start by analyzing the efficiency of \toboggan and \catfish,
noting that this compares a Python implementation~\cite{toboggan} with C++ code~\cite{catfish-github}.
Their runtimes on the 1.4M non-trivial instances in the smaller data sets are shown side-by-side in Figure~\ref{fig:runtimes-all}.
We observe that the two implementations are both quite fast on the vast majority of instances: their
median runtimes are $1.24$ milliseconds (ms) (\toboggan) and $3.47$ms (\catfish).
However, the implementations have different runtime distributions---whereas \catfish is consistent, terminating between
2.3--4.6ms on $90\%$ of instances and never running longer than 1.3 seconds,
\toboggan trades off faster termination, e.g. less than $2$ms on $80\%$ of instances,
with a higher variance and a small chance of a much longer runtime, e.g. over $50$ seconds on $0.48\%$ of instances.

\begin{figure}
  \centering
  \includegraphics[width=0.95\textwidth]{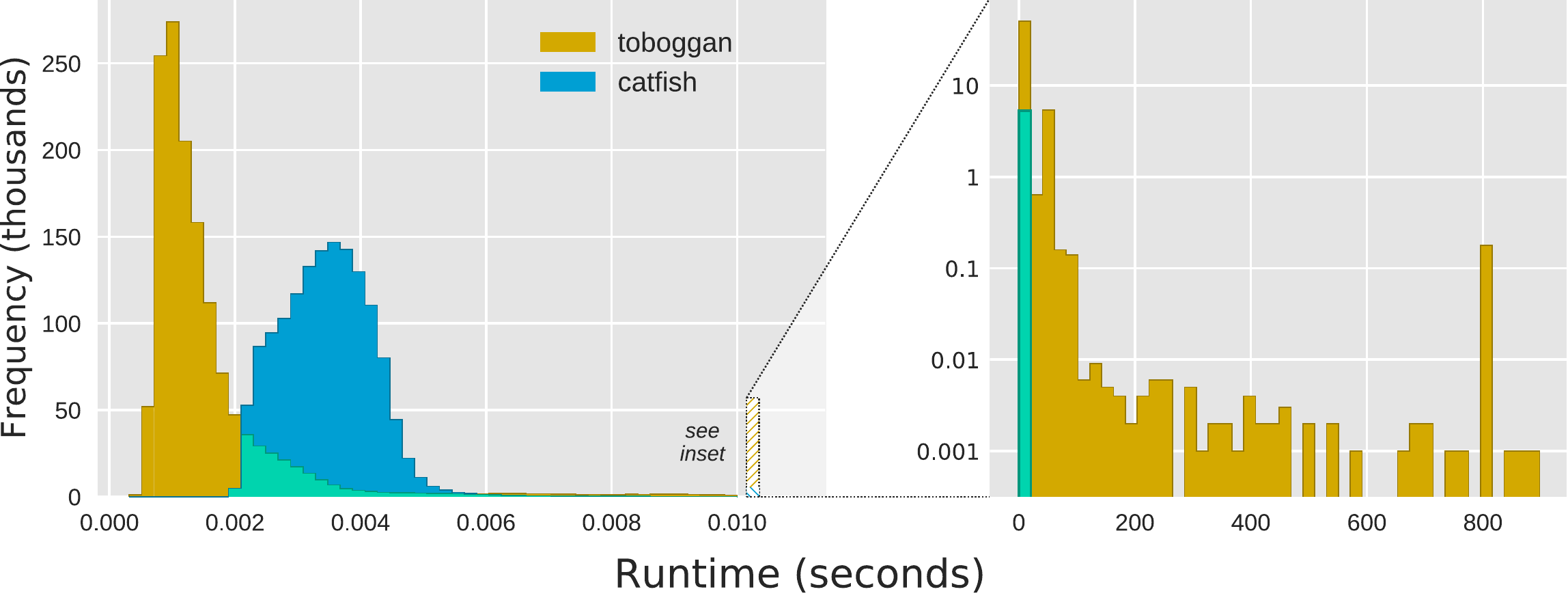}%
  \caption{\label{fig:runtimes-all}
    Runtimes of \toboggan and \catfish on all non-\salmon, non-trivial instances. The $y$-axes indicate the number of instances on which the algorithms terminate in the given time window.
    }
\end{figure}

\subsection{Model Validation}\label{sec:exp:validation}
Previous papers that use flow decompositions to recover RNA
sequences~\cite{Catfish,tomescu2013novel} implicitly assume that the
true RNA sequences correspond to a \emph{minimum} size flow decomposition,
as opposed to one with an arbitrary number of paths.
Because \toboggan provably finds the minimum size of a flow decomposition, our implementation enables us to investigate exactly how often this assumption holds in practice.

Table~\ref{tab:validation} gives the percentage of all non-trivial instances whose ground truth decompositions are in fact minimum decompositions, as well as the percentage of ground truth decompositions we proved have non-optimal size.
We conclude from this table that the \FD problem is in fact
a useful model for transcript assembly, which underscores the need for
efficient algorithms to compute minimum decompositions.

\begin{table}[!ht]
    \centering
    \begin{tabularx}{0.9425\textwidth}{lrrrrrrr}
      \toprule
      dataset &  instances &  non-trivial &  nodes  &  deg & $k$  &  optimal &  non-optimal \\
      \midrule
      \zebra & 1,549,373 & 445,880 & 18.49 & 2.27 & 2.32 &  99.91\% & 0.053\% \\ 
      \mouse & 1,316,058 & 473,185 & 18.67 & 2.37 & 2.75 &  99.40\% & 0.074\% \\ 
      \human & 1,169,083 & 529,523 & 18.79 & 2.41 & 2.83 &  99.49\% & 0.043\% \\ 
      \salmon & 13,300,893 & 7,153,472 & 20.54 & 2.55 & 3.74 &  94.39\% & 0.035\% \\ 
      \midrule
      all & 17,335,407 & 8,602,060 & 20.22 & 2.52 & 3.55 & 95.27\% & 0.039\% \\ 
    \end{tabularx}\vspace*{1em}
    \caption{\label{tab:validation}
        Summary of the RNA sequencing dataset~\cite{catfish-github}.
        Statistics are for non-trivial instances. Columns 4 through 6 give averages; column 7 (8) reports the percent of non-trivial ground truth decompositions that are optimal (non-optimal) size.
        Because \toboggan timed out on some instances, these percentages do not sum to 100.
        The high percentage of instances with ground truth of minimum size
        supports the use of \FD as a model for transcript assembly.
    }
    \vspace*{-0.3in}
\end{table}

\subsection{Ground Truth Recovery}
Though the ground truth flow decompositions are almost always of
minimum size, it is biologically desirable to find a \emph{particular} minimum
size decomposition rather than an arbitrary one.
In this section we investigate how often the decompositions output by \toboggan and \catfish are identical to the ground truth decomposition, restricting our attention to non-trivial instances in which the ground truth decomposition is of minimum size.
Additionally, for those instances where each algorithm does not exactly recover the ground truth, we analyze the similarity of the (imperfect) path set to the ground truth, using the Jaccard index.\looseness-1

The table in Figure~\ref{fig:solution-quality} summarizes the performance of the \toboggan and \catfish implementations in
exactly computing the ground truth decompositions.
Their behavior was quite similar, with slight differences at each decomposition size, and a $0.2\%$ advantage for \toboggan overall.

For instances where an algorithm's output is not identical to the ground truth, an output can still recover some useful information if it is highly similar to the ground truth decomposition.
With this in mind, we evaluate how similar each algorithm's output is to the ground truth when they do not exactly match\footnote{There are 43,817 such instances for \catfish and 41,783 for \toboggan.}.
The plot in Figure~\ref{fig:solution-quality} shows the distribution of the Jaccard index of each algorithm's output compared to the ground truth paths.

\begin{figure}
    \centering
    \begin{minipage}[t]{0.40\textwidth}
        \resizebox{\textwidth}{!}{
             \begin{tabular}{rrccrr}
                 \toprule
                 $k$ & instances  &  &  &  \catfish &  \toboggan  \\
                 \midrule
                 2 &        63.2791\% & & &   0.992 &   \textbf{0.995} \\
                 3 &        22.0775\% & & &   0.967 &   \textbf{0.969} \\
                 4 &         8.5237\% & & &   \textbf{0.931} &   0.930 \\
                 5 &         3.4920\% & & &   0.886 &   0.886 \\
                 6 &         1.5375\% & & &   \textbf{0.830} &   0.828 \\
                 7 &         0.6698\% & & &   \textbf{0.788} &   0.780 \\
                 8 &         0.2889\% & & &   \textbf{0.767} &   0.766 \\
                 9 &         0.1241\% & & &   0.740 &   \textbf{0.743} \\
                10 &         0.0070\% & & &   0.752 &   \textbf{0.802} \\
                11 &         0.0004\% & & &   0.500 &   0.500 \\
                \midrule
                all & 100\% & & &   0.973 &   \textbf{0.975} \\
                \bottomrule
             \end{tabular}
         }
    \end{minipage}
    \hspace{0.15in}
    \begin{minipage}[t]{0.49\textwidth}
        \raisebox{-0.48\height}{\includegraphics[width=\textwidth]{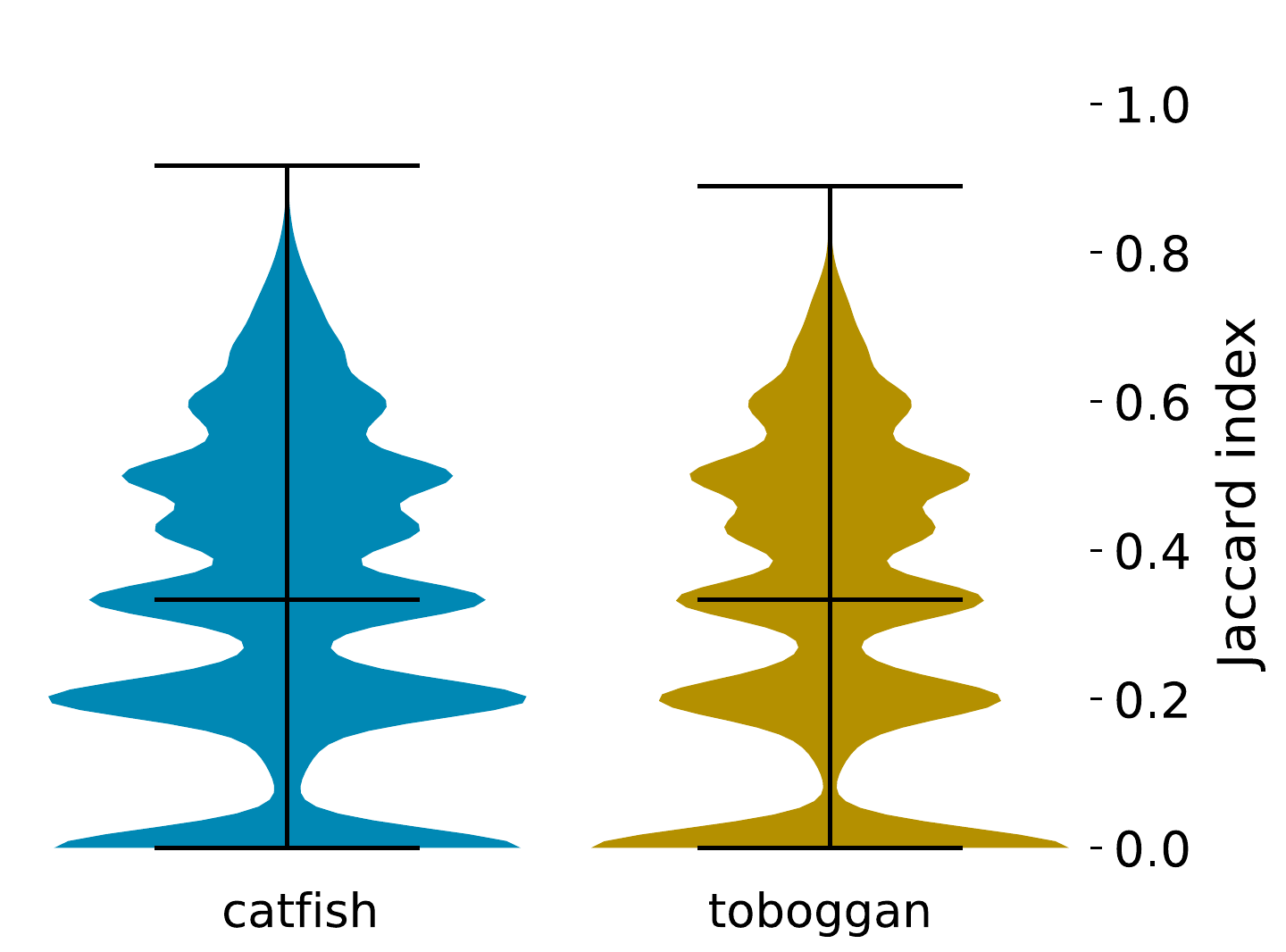}}%
    \end{minipage}%
    \caption{\label{fig:solution-quality}
       (\textit{Left}) Proportion of ground truth path sets that \catfish  and \toboggan recover exactly, organized by path set size ($k$).
       Bold numbers indicate the algorithm with better performance.
       (\textit{Right}) Distributions of the Jaccard index between the algorithms' solutions and the ground truth
       on instances where the paths are not exactly recovered.
   }
\end{figure}

\section{Kernelization lower bounds}\label{sec:no-poly-kernel}

As described in Section~\ref{sec:preprocessing}, our implementation employs a graph reduction algorithm that rapidly identifies any trivial graph and immediately solves \FD.
Out of the 17M total instances, this preprocessing solves all of the roughly 8.7M trivial instances.
On non-trivial instances, \toboggan then tries to identify the correct number $k$ of paths in an optimal solution.
The na\"ive lower bound from the largest edge cut is equal to the correct value of $k$ in $96.5\%$ of the 8.6M non-trivial instances;
incorporating the bound from Lemma~\ref{lem:k_opt_bound} brings this up to $98.4\%$.

In the framework of parameterized complexity, it is therefore natural to ask
whether~\kFD admits a polynomial kernel. Below we answer the question in the
negative, despite the real-world success of our preprocessing. Our proof
strategy requires hardness reductions involving the following restricted
variants of \FD.

\begin{problem}{\UFDfull (\UFD)}
    \Input & An \stdag, $G$ with an integral flow~$f$, an integer $k$, and a set $\weightset \subset \Z$. \\
    \Prob  & Does $(G,f)$ have a flow decomposition into $k$ paths whose weights are all members of $\weightset$?
\end{problem}
\begin{problem}{\WPUFDfull (\WPUFD)}
    \Input & An \stdag, $G$, integral flow $f$,
    and $k$ integers~$\pweights=(w_1, \ldots , w_k)$ taken from $\weightset \subset \Z$. \\
    \Prob & Is there a flow decomposition of $(G,f)$ into $k$ paths with respective weights $\pweights$?
\end{problem}

\begin{lemma}
    \WPxFDfull{1,2,4} is \NP-complete.
\end{lemma}
\begin{proof}
    Consider an instance~$(G,f,k)$ of \xFD{1,2,4}, which is \NP-    complete~\cite{HowToSplit}. Note that every \stpath in a potential
    solution to such an instance can only take weights in~$\{1,2,4\}$. This
    enables us to Turing-reduce~$(G,f,k)$ to at most~$k^3$ instances of
    \WPxFD{1,2,4}: we simply guess how many
    of the~$k$ \stpaths use each of the three possible values.
    It follows that \WPxFD{1,2,4} is \NP-complete.
\end{proof}

\noindent
In order to show that \kFD does not admit polynomial kernels, we
will provide a \emph{cross-composition} from \WPxFD{1,2,4}.
We first need the following technical definition to set up the necessary machinery:

\begin{definition}[Polynomial equivalence relation~\cite{CrossComp}]
    An equivalence relation $\mathcal R$ over~$\Sigma^*$ is called a \emph{polynomial equivalence relation}
    if the following hold:
    \begin{enumerate}
        \item there exists an algorithm that decides for~$x,y \in \Sigma^*$ whether $x$ and~$y$ are
              equivalent under~$\mathcal R$ in time polynomial in~$|x| + |y|$, and
        \item for any finite set~$S \subseteq \Sigma^*$ the index~$| S / \mathcal R |$ is
              bounded by a polynomial in~$\max_{x \in S} |x|$.
    \end{enumerate}
\end{definition}

\noindent
The benefit of a polynomial equivalence relation is that we can focus on
collection of instances which share certain characteristics, as long
as these characteristics do not distinguish too many instances. A simple example is
that we can ask for input instances~$(G_i,f_i,k_i)$ in which all graphs~$G_i$ have
the same number of vertices and the values~$k_i$ are the same.

\begin{definition}[AND-cross-composition~\cite{CrossComp}]\label{def:and-cross-comp}
    Let~$L \subset \Sigma^*$ be a language, $\mathcal R$ a polynomial
    equivalence relation over~$\Sigma^*$, and let~$\Pi \subseteq \Sigma^* \times \N$
    be a parameterized problem. An AND-cross composition of~$L$ into~$\Pi$ (under~$\mathcal R$) is an
    algorithm that, given $\ell$ instances~$x_1,\ldots x_t \in \Sigma^*$ of~$L$ belonging to the
    same equivalence class of~$\mathcal R$, takes time polynomial in~$\sum_{i=1}^\ell |x_i|$ and
    outputs an instance~$(y,k) \in \Sigma^* \times \N$ such that
    \begin{enumerate}[label=\alph*)]
        \item the parameter~$k$ is polynomially bounded in~$\max_{1 \leq i \leq \ell}|x_i| + \log \ell$, and
        \item we have that~$(y,k) \in \Pi$ if and only if \emph{all} instances~$x_i \in L$.
    \end{enumerate}
\end{definition}

\noindent
We will now use the following theorem (abridged to our needs here) and the subsequent
AND-cross-composition to prove that \kFD does not admit small kernels.

\begin{proposition}[Bodlaender, Jansen, Kratsch~\cite{CrossComp}]\label{thm:no-poly-kernel}
    If an \NP-hard language~$L$ AND-cross-composes into a parameterized problem~$\Pi$,
    then $\Pi$ does not admit a polynomial kernelization unless $\NP \subseteq \coNP / \emph{poly}$
    and the polynomial hierarchy collapses.
\end{proposition}

\begin{theorem}\label{thm:kfd_nopolyker}
    \kFDfull does not admit
    a polynomial kernel unless $\NP \subseteq \coNP / \emph{poly}$ and the polynomial hierarchy collapses.\looseness-1
\end{theorem}
\begin{proof}
    Let $\mathcal R_w$ be the equivalence relation on instances of \WPxFD{1,2,4}
    where $(G_1,f_1, \pweights_1) \equiv (G_2,f_2,\pweights_2)$ if and only if
    $\pweights_1 = \pweights_2$. Since each entry of $\pweights_i$ is in
    $\{1,2,4\}$, $\mathcal R_w$ has at most $k^3$ equivalence classes, and is
    a polynomial equivalence relation.

    Let~$x_1,\ldots,x_\ell$ be instances of \WPxFD{1,2,4} all contained
    in the same equivalence class of~$\mathcal R_w$, with
    $x_i = (G_i,f_i,\mathbf{w})$ and $\mathbf{w} = (w_1, \ldots, w_k)$
    the common prescribed path weights. We denote the source and sink
    of $G_i$ by $s_i$ and $t_i$, respectively.

    We construct an additional instance~$x_{\ell +1}$ as follows:
    $G_{\ell+1}$ consists of two vertices~$s_{\ell+1}, t_{\ell+1}$, and
    $k$ arcs~$a_1,\ldots,a_k$ from $s_{\ell+1}$ to $t_{\ell+1}$.
    The flow~$f_{\ell+1}$ has value $w_i$ on arc $a_i$.
    By construction~$x_{\ell+1} = (G_{\ell+1}, f_{\ell+1}, \mathbf{w})$
    is a positive instance of \WPxFD{1,2,4};
    moreover, it has a unique decomposition into $k$ \stpaths (up to isomorphism).

    Before we describe the composition, we treat a technicality.
    If the total flow~$F_i$ for any $f_i$ is different from~$\sum_{j = 1}^k w_j$,
    then $x_i$ is a negative instance. In this case, instead of a composition
    we output a trivial negative instance $y$ for \kFD.
    Otherwise, we compose the instances~$x_1,\ldots,x_{\ell+1}$ into a single
    composite instance~$y = (G,f,k)$ of \kFD.
    To form $G$, we chain the $G_i$s together
    by identifying the vertex~$t_i$ with the vertex~$s_{i+1}$ for~$1 \leq i \leq \ell$.
    The resulting $G$ is an \stdag with source~$s = s_1$ and sink~$t = t_{\ell+1}$.
    We define $f$ to label each arc in $G$ with the flow value from its original instance.
    Since each $x_i$ has the same total flow, $f$ is a flow on $G$.

    We point out that property a) of Definition~\ref{def:and-cross-comp} is trivially satisfied.
    To see that property b) holds, first assume that all instances~$x_i$ are positive.
    Since all these solutions consist of $k$ \stpaths with the same prescribed values~$w_1,\ldots,w_k$,
    we can chain the individual flow decompositions together into~$k$ \stpaths in~$G$. Accordingly,
    $y$ is then a positive instance. In the other direction, assume that $y$ has a solution, \ie~$f$
    can be split into exactly~$k$ \stpaths. Due to our inclusion of the instance~$x_{\ell+1}$ in
    the construction of~$y$, the respective values of the \stpaths must be exactly~$w_1,\ldots,w_k$.
    But then restricting this global solution to each individual instance~$x_i$ (since all \stpaths meet
    at the identified source/sink cut vertices) produces a solution. We conclude that therefore all~$x_i$
    must have been positive instances, as required

    Finally, our construction clearly takes time polynomial in~$\sum_{i=1}^\ell |x_i|$ making
    it an AND-cross-composition of \WPxFD{1,2,4} into
    \kFD. Invoking Theorem~\ref{thm:no-poly-kernel}, this proves that
    \kFD does not admit a polynomial kernel unless $\NP \subseteq \coNP / \emph{poly}$.
\end{proof}

\noindent
We note that our construction in the proof of Theorem~\ref{thm:kfd_nopolyker} produces
an instance of \UFD, which can easily be Turing reduced into an instance of \WPUFD.

\begin{corollary}
    Unless $\NP \subseteq \coNP / \emph{poly}$, the problems \UFDfull and \WPUFDfull
    do not admit polynomial kernels.
\end{corollary}

\section{Hardness of assigning weights}\label{sec:exact3}

Every solution computed by \toboggan in our experiments corresponded to a fully
determined linear system in the final dynamic programming table. This means
that we never had to run the expensive ILP solver to determine the weights;
instead, they were computed in polynomial time with respect to~$k$ using row
reduction. This raises the question: is the linear system obtained from a
decomposition into paths always fully determined? In the following we show that the
answer must be `no' in case of \emph{non-optimal} decompositions. Not only
can there be multiple weight-assignments for the same set of paths,
recovering these weights is actually \NP-hard.
Formally, we consider the problem \PPFDfull.

\begin{problem}{\PPFDfull (\PPFD)}
    \Input & An \stdag~$G$ with an integral flow~$f$,
             and a prescribed set of \stpaths, $P = p_1, \cdots, p_k$.\\
    \Prob  & Identify integral weights~$\pweights = (w_1,\dots w_k)$
             such that $(P, \pweights)$ is a flow decomposition of $G$.
\end{problem}

\noindent
Our proof that \PPFD is \NP-complete relies on a reduction from \XHSfull,
which is equivalent to monotone 1-in-3-SAT and known to be \NP-hard~\cite{1in3SatHardness}.

\begin{problem}{\XHSfull (\XHS)}
    \Input & A finite universe $U = \{u_1, \cdots, u_n \}$,
    and a collection $\mathcal{S} \subseteq \binom{U}{3}$ of subsets of $U$ of size 3.  \\
    \Prob  &  Find a subset $X \subseteq U$ such that
    each element of $\mathcal{S}$ intersects (``hits'') $X$ \emph{exactly} once,
    i.e. $\forall S \in \mathcal{S}, |S\cap X| = 1$.
\end{problem}

 \begin{theorem}\label{thm:x3h-reduction}
     \PPFDfull is \NP-complete.
 \end{theorem}

\begin{proof}
  Given an instance $(U, \mathcal S)$ of \XHS, we will construct an equivalent instance of \PPFD.
  Our approach is to encode each element of $U$ using two ``partner'' \stpaths in an \stdag, $G$,
  whose weights sum to 3, and each triad in $\mathcal{S}$ as a different \TOcut in $G$.
  We will route the partner paths and assign flow values so that the set of \stpaths
  with weight 2 exactly correspond to elements in a solution to the hitting set problem.

  We first construct $G$. Let $V(G) = \{s, v_1, \ldots, v_{|\mathcal{S}|}, t = v_{|\mathcal{S}|+1}\}$,
  where $v_i$ is associated with the $i$th triad in $\mathcal{S}$ for $1 \leq i \leq |\mathcal{S}|$,
  and $s$/$t$ are source/sink vertices. Our construction will be such that each vertex other than $t$
  has exactly one out-neighbor: $s$ has out-neighbor $v_1$ and $v_i$ has out-neighbor $v_{i+1}$.
  Create one $sv_1$ arc with flow value 3 for each element of $U$, and
  $(|U|-1)$ $v_iv_{i+1}$ arcs---one each of flow values 4 and 5, and $|U|-3$ with flow value 3.
  We now define a set of prescribed paths $P = \{p_Y, \bar{p}_Y\}_{Y\in U}$.
  For each element $Y \in U$, the corresponding partner \stpaths $p_Y$ and $\bar{p}_Y$ are
  routed over the same arc out of $s$. This guarantees that each pair of partner paths must receive weights
  summing to 3. At $v_i$, we route these paths to encode the corresponding triad $S_i = \{u_1, u_2, u_3\}$
  as follows. The paths $p_{u_1}$, $p_{u_2}$, and $p_{u_3}$ are routed over the arc with flow 4
  and $\bar p_{u_1}$, $\bar p_{u_2}$, and $\bar p_{u_3}$ go over the arc of flow 5.
  Now assign each element $u' \in U\setminus S_i$ to a unique $v_iv_{i+1}$ arc $a_{u'}$ of flow value 3
  and route $p_{u'}$ and $\bar p_{u'}$ together over $a_{u'}$.
  This construction is illustrated in Figure~\ref{fig:x3h-reduction}.

  To complete the proof, we describe how a solution to this instance of \PPFD yields
  a solution to the original instance of \XHS. Consider the possible values of the weights
  of the paths in $P$. By design, the out-arcs of $s$ force each pair of partner paths
  to have one of weight 1 and one of weight 2.
  Each triad $S_i$ is represented by two arcs out of $v_i$: one with flow value 4 ($a_i$),
  and the other with 5 ($\bar{a_i}$). Because exactly three prescribed \stpaths
  are routed over $a$, and all our paths must have weight 1 or 2, this
  guarantees that exactly one \stpath routed over $a$ has weight 2 (and the other two must have weight 1).
  Thus, finding a set of weights that solve this instance of \PPFD is equivalent to
  choosing a set of paths to have weight 2 such that exactly one selected path
  is routed over each arc with flow value 4. But this is equivalent to choosing a set of elements
  (paths) such that each triad (arc) is hit by (incident to) exactly one of the chosen elements (\stpaths of weight 2).
  \looseness-1
\end{proof}

\begin{figure}[t]
    \centering
    \includegraphics[scale=.7]{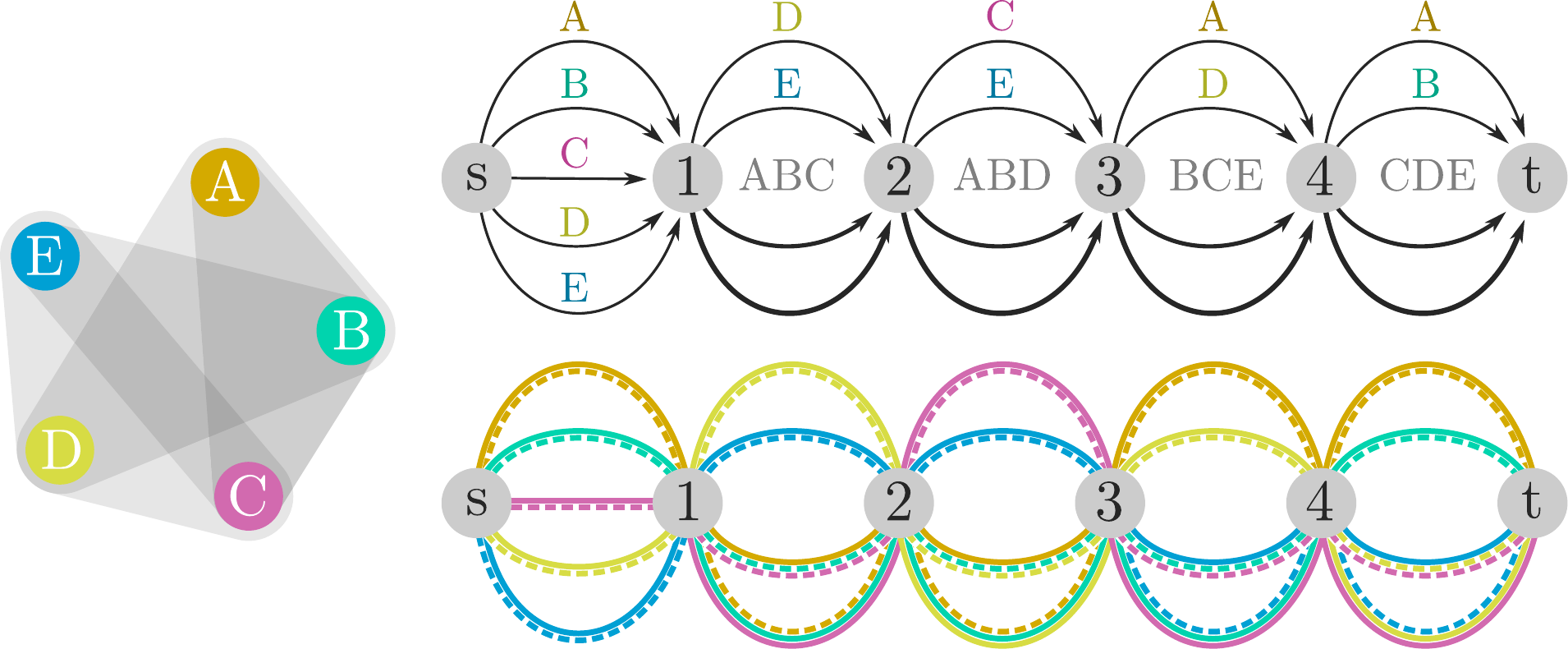}
    \caption{\label{fig:x3h-unique}\label{fig:x3h-reduction}
        (\textit{Left})
        Instance of \XHS with unique solution $\{A,E\}$.
        (\textit{Right})
        The \stdag from our reduction, and a flow decomposition corresponding to the solution~$\{A,E\}$.
        Light arcs have weight~$3$, medium~$4$, and heavy~$5$; dashed paths have weight~$1$ and solid paths~$2$.
    }
\end{figure}

\noindent
In the above reduction, we looked for a decomposition with $2\cdot|U|$ prescribed paths.  However, if we do not prescribe paths, there is a flow decomposition with $|U|+1$ paths
($|U|-1$ of weight three, one of weight two, and one of weight one) and a
corresponding linear system of full rank\footnote{This is true regardless of
whether $(U,\mathcal{S})$ is a a yes-instance.}. As such, it is possible that
\PPFD is not difficult when the prescribed paths on a yes-instance are part of an optimal decomposition.

\begin{conjecture}\label{conjecture}
  If $k$ is the minimum value for which $(G,f)$ has a flow decomposition of
  size $k$, then every integer-weighted solution has a corresponding linear
  system $L$ of rank $k$.
\end{conjecture}

\noindent
A direct consequence of this conjecture would be that the running time in
Theorem~\ref{thm:runtime} immediately improves to~$\frac{1}{k!}\cdot 4^{k^2}
0.649^k k^{O(1)} \cdot n$ since the ILP-solving step would never occur. In
this context, we note that Vatinlen et al.~\cite{greedy} proved that when $k$
is minimum, every solution with real-valued weights has a corresponding linear
system of full rank. However, their proof does not hold when the path weights
are restricted to the integers.

\section{Conclusion}
\label{sec:conclusion}

We presented a holistic treatment of \FD from the perspectives of parameterized complexity
and algorithm engineering, resulting in a competitive solver, \toboggan.
Our approach verifies that parameterized algorithms can (with sufficient engineering)
be applied to real-world problems even in high-throughput situations.
Our work also naturally leads to several practical and theoretical questions for further investigation.

On the practical side, we would like to understand the cases in which a minimal
flow decomposition does not match the assembly problem's ground truth, and how
we might improve the recovery rate. The similarity in performance of \toboggan and
\catfish in our experiments suggests that we need to refine either the problem formulation
or our notion of minimality; in either case, more domain-specific knowledge is
needed.

On the theory side, we ask whether there exists an fpt algorithm for \FD with
running time~$k^{O(k)} n$ or better. In particular, it will be interesting to
see whether the established techniques used to improve dynamic programming
algorithms~\cite{SingleExpTWDet,CutAndCount} are applicable to our algorithm.
Furthermore, if Conjecture~\ref{conjecture} holds, it immediately implies
a tighter upper bound on the running time of our algorithm,
and might lead to further optimizations.\looseness-1

\vspace*{0.05in}
\noindent\textbf{Acknowledgements}.
{\small This work supported in part by the
Gordon \& Betty Moore Foundation's Data-Driven Discovery Initiative through
Grant GBMF4560 to Blair D.~Sullivan.}

\bibliographystyle{abbrv}
\bibliography{./biblio}
\vfill

\appendix
\section{Pseudocode for \kFDfull}\label{sec:pseudo}
\begin{algorithm}[H]
\label{alg:main}
\SetKwInOut{Input}{input}\SetKwInOut{Output}{output}

 \Input{An \stdag $G$, a flow $f$, and an integer $k$.}
 \Output{A vector $\pweights$ that contains the weights of the \stpaths for a flow decomposition into $k$ integral-weighted \stpaths. Or $\emptyset$ if none exist.}

  $order(G)$ \tcp*[r]{Order vertices via topological ordering}
  \tcp{Build $T_0$}
  \For{$i\in [k]$}{
    $g_0(i) := a_s$\;
  }
  $L_0 := \Big[\sum_{i=1}^{k} w_i = F\Big]$\;
  $T_0 := (g_0, L_0)$\;

  \tcp{Do iterative steps}
  \For{i=1 \KwTo n}{
    $T_i = \emptyset$\;
    \For{$(g,L) \in T_{i-1}$}{
      \ForAll{extensions $g'$ of $g$}{
        $L' := L$\;
        \For{$a \in A^+(v_i)$}{
          add equation $\Big[\sum_{i \in g'^{-1}(a)} w_i = f(a)\Big]$ to $L'$\;
        }
        $reduce(L')$ \tcp*[r]{Remove linearly dependent rows}
        \If{$L'$ is consistent}{
          $T_i = T_i \cup \{(g',L')\}$\;
        }

      }
    }
    \If{ $T_i = \emptyset$}{
    \Return{$\emptyset$}\;
  }
  }

  \tcp{Find $L$ in final table that has an integer solution}
  \For{$(g,L) \in T_n$}{
    \If{$L$ has an integer solution $\mathbf w$}{
      \Return{$\mathbf w$}\;
    }
  }
  \tcp{No solution was found}
  \Return{$\emptyset$}\;
 \vspace*{1em}
 \caption{Linear-fpt algorithm for deciding \kFDfull.}
\end{algorithm}

\newpage
\section{Experimental results organized by species}\label{sec:appendix-recovery}

Following the experimental setup of~\cite{Catfish}, in this section we report our results on each of the species data sets individually.
Figures~\ref{fig:appendix-recovery-by-species}~and~\ref{fig:appendix-failures-overlap} give this breakdown for the aggregated results shown in Figure~\ref{fig:solution-quality}.

\vspace{0.5in}
\begin{figure}[h]
    \centering
    \captionsetup[subfigure]{justification=centering}
    \adjustbox{valign=t}{
        \begin{subfigure}{0.32\textwidth}
                \resizebox{\textwidth}{!}{
                    \begin{tabular}{rrrr}
                        \toprule
                        $k$ & instances  &  \catfish &  \toboggan  \\
                        \midrule
                          2 &        76.6132\% &   0.992 &   \textbf{0.995} \\
                          3 &        17.3138\% &   0.962 &   \textbf{0.963} \\
                          4 &         4.3831\% &   \textbf{0.913} &   0.911 \\
                          5 &         1.1359\% &   0.855 &   \textbf{0.858} \\
                          6 &         0.3731\% &   \textbf{0.765} &   0.761 \\
                          7 &         0.1174\% &   \textbf{0.700} &   0.696 \\
                          8 &         0.0411\% &   0.710 &   \textbf{0.727} \\
                          9 &         0.0157\% &   \textbf{0.700} &   0.643 \\
                         10 &         0.0054\% &   \textbf{0.833} &   0.792 \\
                         11 &         0.0013\% &   0.500 &   0.500 \\
                        \midrule
                         All & 100\% &    0.980 &   \textbf{0.983} \\
                    \end{tabular}
                }%
            \caption{\zebra}
            \label{fig:appendix-recovery-zebra}
        \end{subfigure}%
    }%
    \adjustbox{valign=t}{
        \begin{subfigure}{0.32\textwidth}
                \resizebox{\textwidth}{!}{
                    \begin{tabular}{rrrr}
                        \toprule
                        $k$ & instances  &  \catfish &  \toboggan  \\
                        \midrule
                          2 &        59.4943\% &   0.992 &   \textbf{0.995} \\
                          3 &        23.4974\% &   0.966 &   \textbf{0.968} \\
                          4 &         9.6369\% &   \textbf{0.930} &   0.928 \\
                          5 &         4.1312\% &   0.880 &   \textbf{0.883} \\
                          6 &         1.8605\% &   \textbf{0.821} &   0.814 \\
                          7 &         0.8402\% &   0.776 &   \textbf{0.769} \\
                          8 &         0.3755\% &   \textbf{0.773} &   0.770 \\
                          9 &         0.1573\% &   \textbf{0.751} &   0.746 \\
                         10 &         0.0066\% &   0.677 &   \textbf{0.742} \\
                         & & & \\
                        \midrule
                         All & 100\% &    0.969 &   \textbf{0.971} \\
                    \end{tabular}
                }%
            \caption{\mouse}
            \label{fig:appendix-recovery-mouse}
        \end{subfigure}%
    }%
    \adjustbox{valign=t}{
        \begin{subfigure}{0.32\textwidth}
                \resizebox{\textwidth}{!}{
                    \begin{tabular}{rrrr}
                        \toprule
                        $k$ & instances  &  \catfish &  \toboggan  \\
                        \midrule
                          2 &        55.3832\% &   0.992 &   \textbf{0.996} \\
                          3 &        24.8378\% &   0.970 &   \textbf{0.973} \\
                          4 &        11.0312\% &   0.937 &   \textbf{0.939} \\
                          5 &         4.9135\% &   \textbf{0.897} &   0.894 \\
                          6 &         2.2336\% &   0.846 &   \textbf{0.847} \\
                          7 &         0.9848\% &   \textbf{0.805} &   0.798 \\
                          8 &         0.4212\% &   0.766 &   \textbf{0.767} \\
                          9 &         0.1860\% &   0.734 &   \textbf{0.748} \\
                         10 &         0.0087\% &   0.761 &   \textbf{0.848} \\
                         & & & \\
                        \midrule
                         All & 100\% &    0.969 &   \textbf{0.972} \\
                    \end{tabular}
                }%
            \caption{\human}
            \label{fig:appendix-recovery-human}
        \end{subfigure}%
    }%
    \caption{Proportion of ground truth pathsets of a given size that \catfish  and \toboggan recover exactly, organized by species. Bolded numbers indicate the algorithm with better performance. These results are reported in aggregate in Figure~\ref{fig:solution-quality}. \label{fig:appendix-recovery-by-species}}
\end{figure}

\begin{figure}[h]
    \centering
    \captionsetup[subfigure]{justification=centering}
    \begin{subfigure}{0.32\textwidth}
        \hspace*{-.6em}\includegraphics[scale=.31]{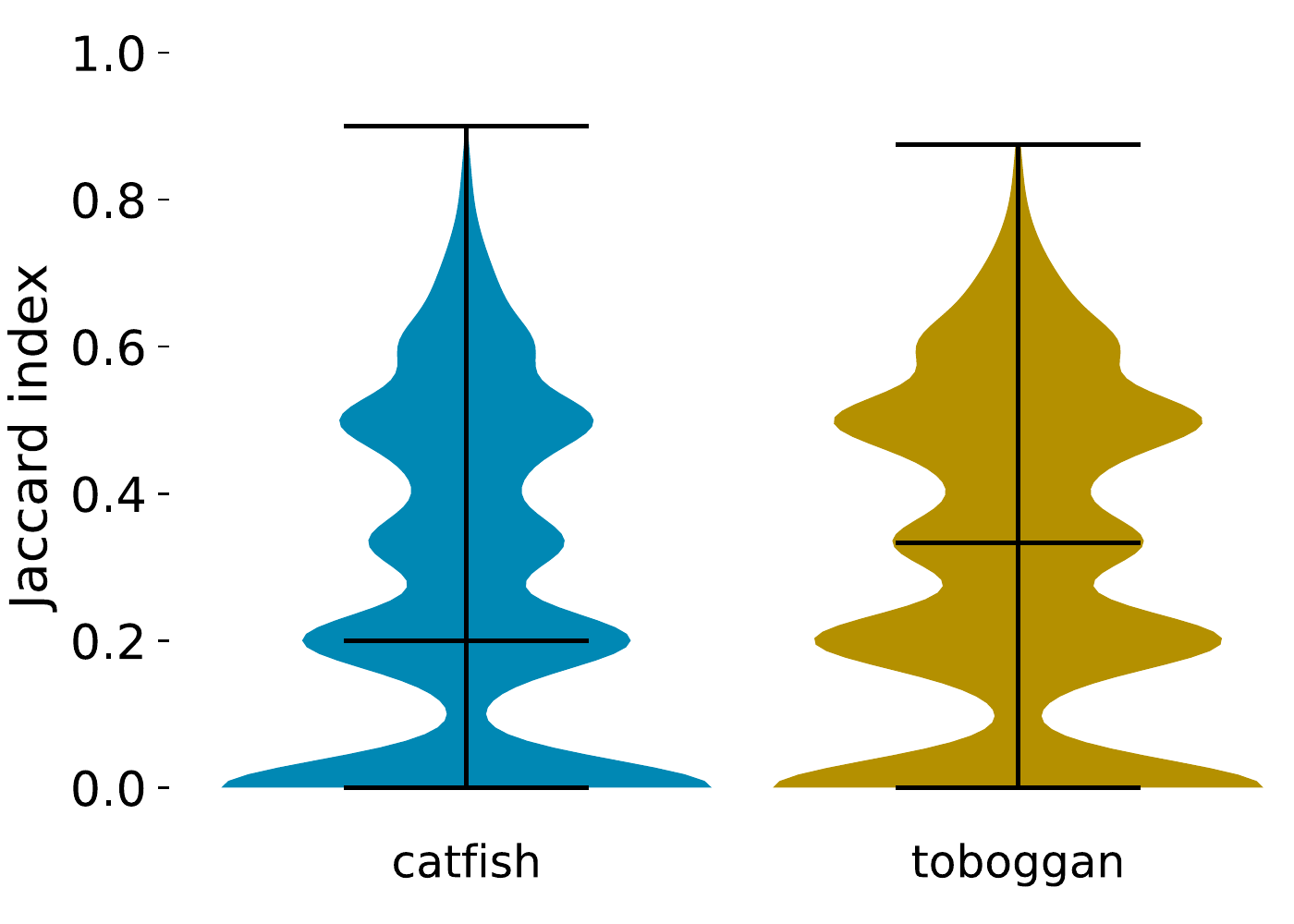}%
        \caption{\zebra}
    \end{subfigure}\hfil
    \begin{subfigure}{0.32\textwidth}
        \hspace*{.3em}\includegraphics[scale=.31]{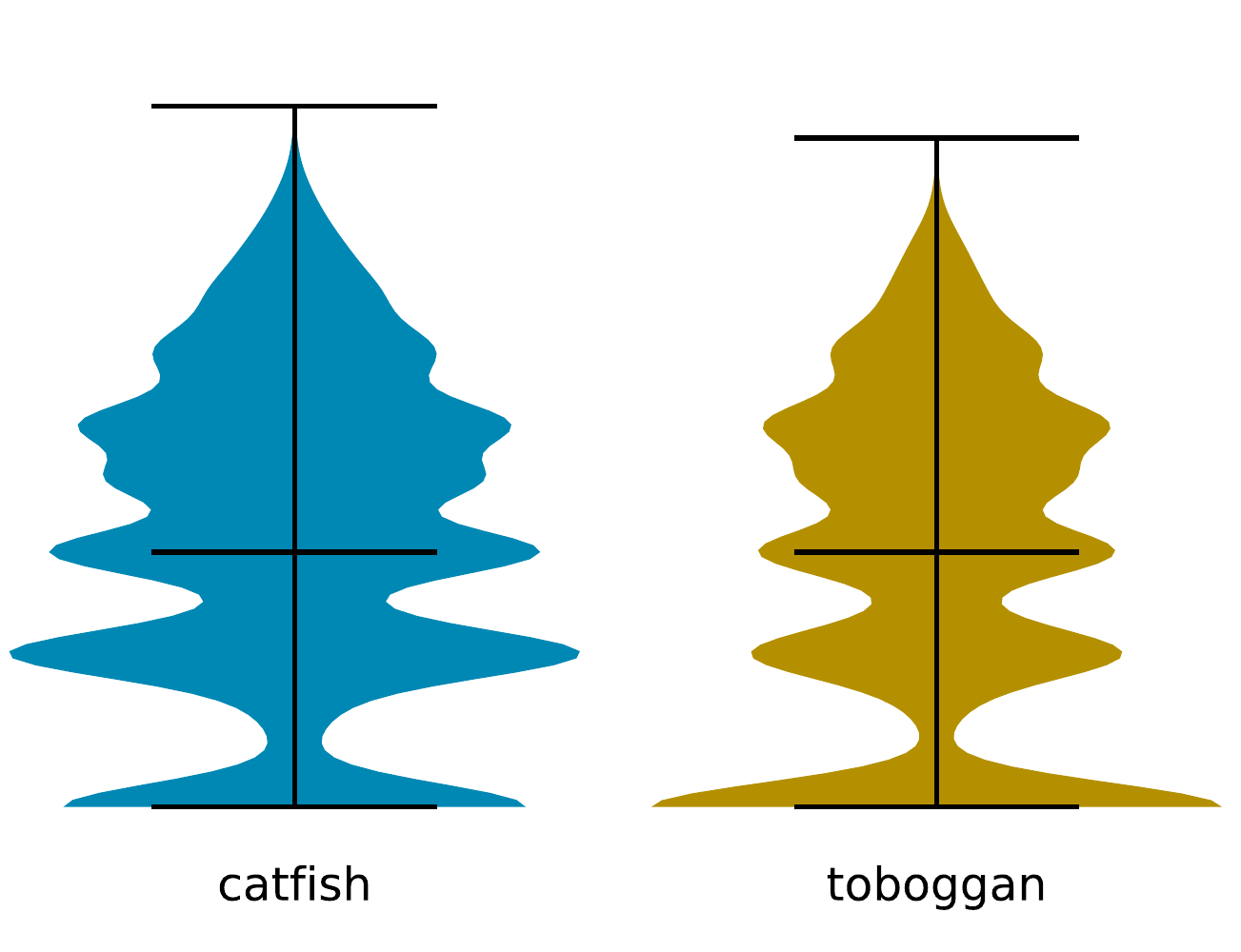}%
        \caption{\mouse}
    \end{subfigure}\hfil
    \begin{subfigure}{0.32\textwidth}
        \hspace{.5em}\includegraphics[scale=.31]{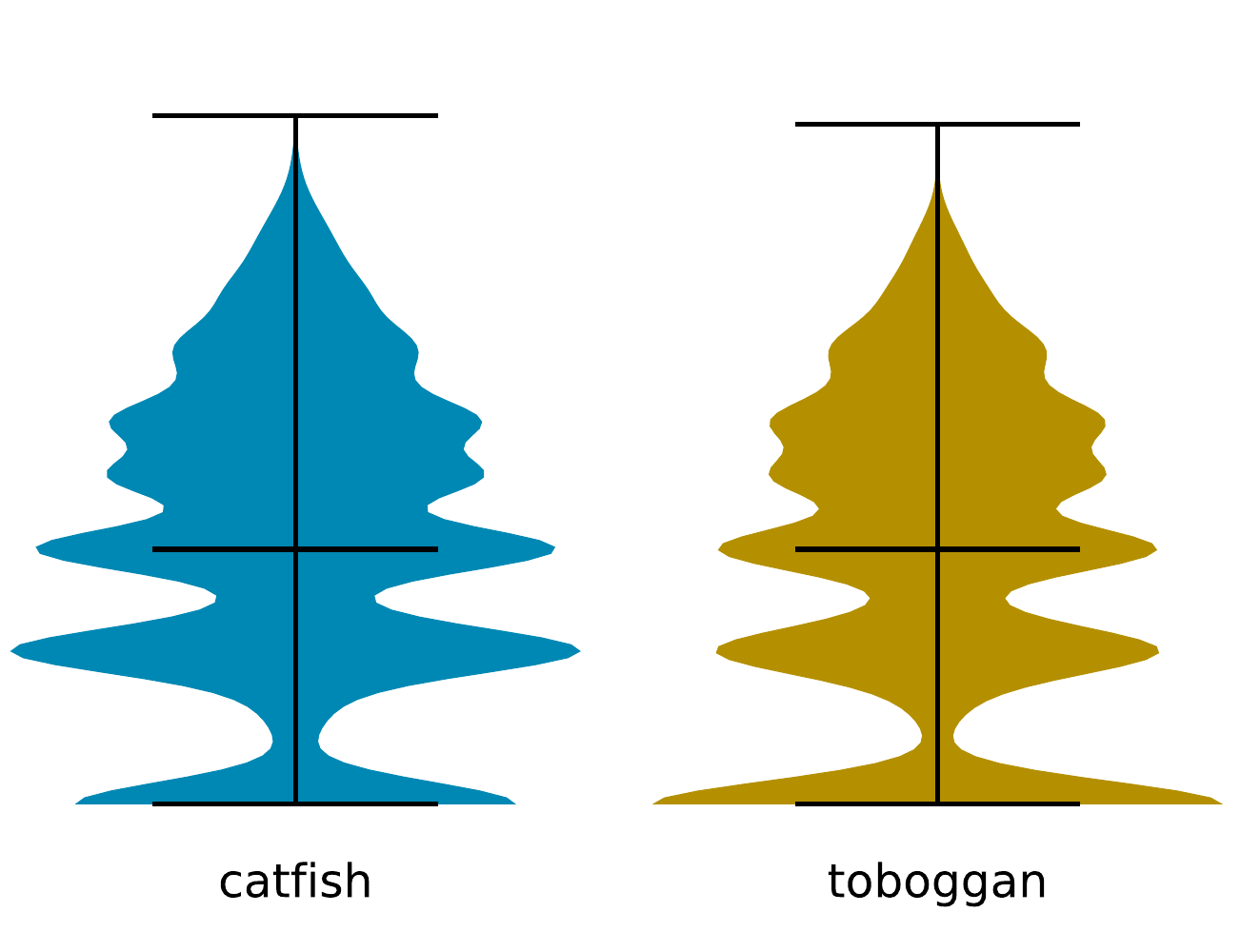}%
        \caption{\human}
    \end{subfigure}
    \caption{\label{fig:appendix-failures-overlap}
    Distribution of Jaccard indices between the algorithms' solutions and the ground truth. The bars in the middle indicate the median value; those at the top and bottom demarcate the extreme values. Instances for which the Jaccard index is 1 are omitted because those statistics are summarized in the tables in Figure~\ref{fig:appendix-recovery-by-species}. These results are reported in aggregate in Figure~\ref{fig:solution-quality}.
    }
\end{figure}

\def\redefineme{
    \def\LNCS{LNCS}%
    \def\ICALP##1{Proc. of ##1 ICALP}%
    \def\FOCS##1{Proc. of ##1 FOCS}%
    \def\COCOON##1{Proc. of ##1 COCOON}%
    \def\SODA##1{Proc. of ##1 SODA}%
    \def\SWAT##1{Proc. of ##1 SWAT}%
    \def\IWPEC##1{Proc. of ##1 IWPEC}%
    \def\IWOCA##1{Proc. of ##1 IWOCA}%
    \def\ISAAC##1{Proc. of ##1 ISAAC}%
    \def\STACS##1{Proc. of ##1 STACS}%
    \def\IWOCA##1{Proc. of ##1 IWOCA}%
    \def\ESA##1{Proc. of ##1 ESA}%
    \def\WG##1{Proc. of ##1 WG}%
    \def\LIPIcs##1{LIPIcs}%
    \def\LIPIcs{LIPIcs}%
    \def\LICS##1{Proc. of ##1 LICS}%
}

\end{document}